\documentclass[a4paper,UKenglish,cleveref, autoref]{lipics-v2019}

\usepackage{amssymb,boxedminipage,amsmath,subcaption,complexity}
\usepackage[normalem]{ulem}
\usepackage[export]{adjustbox}

\usepackage{enumitem}

\usepackage{tikz}
 \usetikzlibrary{calc}
 \usepackage{boxedminipage}
\usepackage{framed}
\usepackage{thmtools}
\usepackage{thm-restate}
\usepackage{xspace}
\usepackage{xifthen}
 \usepackage{tabularx}

{\normalfont\bfseries}{\itshape}

\setlist[itemize]{noitemsep}
\setlist[enumerate]{noitemsep}

\newcommand\yes{\textsc{Yes}}
\newcommand\no{\textsc{No}}

\newcommand\vc{\ensuremath{\mathsf{vc}}\xspace}
\newcommand\fvs{\ensuremath{\mathsf{fvs}}\xspace}
\newcommand\oct{\ensuremath{\mathsf{oct}}\xspace}

\newlength{\RoundedBoxWidth}
\newsavebox{\GrayRoundedBox}
\newenvironment{GrayBox}[1]%
   {\setlength{\RoundedBoxWidth}{.93\textwidth}
    \def\boxheading{#1}
    \begin{lrbox}{\GrayRoundedBox}
       \begin{minipage}{\RoundedBoxWidth}}%
   {   \end{minipage}
    \end{lrbox}
    \begin{center}
    \begin{tikzpicture}%
       \node(Text)[draw=black!20,fill=white,rounded corners,%
             inner sep=2ex,text width=\RoundedBoxWidth]%
             {\usebox{\GrayRoundedBox}};
        \coordinate(x) at (current bounding box.north west);
        \node [draw=white,rectangle,inner sep=3pt,anchor=north west,fill=white]
        at ($(x)+(6pt,.75em)$) {\boxheading};
    \end{tikzpicture}
    \end{center}}

\newenvironment{defproblemx}[2][]{\noindent\ignorespaces%
                                \FrameSep=6pt%
                                \parindent=0pt%
                \vspace*{-1.5em}
                \ifthenelse{\isempty{#1}}{%
                  \begin{GrayBox}{\textsc{#2}}%
                }{%
                  \begin{GrayBox}{\textsc{#2} parameterized by~{#1}}%
                }
                \begin{tabular*}{\textwidth}{@{\hspace{.1em}} >{\itshape} p{1.8cm} p{0.8\textwidth} @{}}%
            }{
                \end{tabular*}%
                \end{GrayBox}%
                \ignorespacesafterend
            }

\newcommand{\defproblema}[3]{%
  \begin{defproblemx}{#1}
    Input:  & #2 \\
    Question: & #3
  \end{defproblemx}
}%

\newcommand\problemdef[3]{
	\defproblema{#1}{#2}{#3}
}

\usepackage{xcolor,cite,xspace}
\usepackage[ruled,vlined,boxed,commentsnumbered,linesnumbered]{algorithm2e}
\definecolor{dark-red}{rgb}{0.4,0.15,0.15}
\definecolor{dark-blue}{rgb}{0.15,0.15,0.4}
\definecolor{medium-blue}{rgb}{0,0,0.5}
\definecolor{gray}{rgb}{0.5,0.5,0.5}
\definecolor{color-Ig}{rgb}{0.15,0.7,0.15}

\hypersetup{
colorlinks = true,
urlcolor={medium-blue},
linkcolor = black!30!blue,
citecolor = black!30!green
}

\newcommand{\ETH}{{\sf ETH}\xspace}
\renewcommand{\FPT}{{\sf FPT}\xspace}
\renewcommand{\XP}{{\sf XP}\xspace}
\renewcommand{\coNP}{{\sf co-NP}\xspace}
\renewcommand{\NP}{{\sf NP}\xspace}
\renewcommand{\P}{{\sf P}\xspace}
\newcommand{\MSO}{{\sf MSO}$_2$\xspace}

\newcommand{\Acal}{\mathcal{A}}

\newcommand{\Hcal}{\mathcal{H}}

\newcommand{\Mcal}{\mathcal{M}}

\newcommand{\Ocal}{\mathcal{O}}

\newcommand\bc{\ensuremath{\mathsf{bc}}\xspace}
\newcommand\tw{\ensuremath{\mathsf{tw}}\xspace}

\newcommand\contracpidk{\textsc{$k$-Contraction($\pi,d$)}\xspace}
\newcommand\contracpid{\textsc{Contraction($\pi,d$)}\xspace}
\newcommand\contracpik{\textsc{$k$-Contraction($\pi$)}\xspace}
\newcommand\contracpi{\textsc{Contraction($\pi$)}\xspace}

\newcommand\contracvc{\textsc{Contraction($\vc$)}\xspace}

\newcommand\contracpione{\textsc{$1$-Contraction($\pi,1$)}\xspace}
\newcommand\contracvcone{\textsc{$1$-Contraction($\vc,1$)}\xspace}
\newcommand\contracvcd{\textsc{Contraction($\vc,d$)}\xspace}
\newcommand\contracfvs{\textsc{Contraction($\fvs$)}\xspace}

\newcommand\contracHcal{\textsc{Contraction($\tau^{\prec}_{\Hcal}$)}\xspace}
\newcommand\contracHone{\textsc{$1$-Contraction($\tau^{\prec}_H,1$)}\xspace}
\newcommand\contracHcalone{\textsc{$1$-Contraction($\tau^{\prec}_{\Hcal},1$)}\xspace}
\newcommand\contracHcalonebullet{\textsc{$1$-Contraction($\tau^{\prec}_{\Hcal^{\bullet}},1$)}\xspace}

\newcommand{\defproblemaOPT}[3]{%
  \begin{defproblemx}{#1}
    Input:  & #2 \\
    Output: & #3
  \end{defproblemx}
}%

\newcommand\problemOPTdef[3]{
	\defproblemaOPT{#1}{#2}{#3}
}

\newcommand\mincontracpi{\textsc{Min-Contraction($\pi$)}\xspace}
\newcommand\mincontracvc{\textsc{Min-Contraction($\vc$)}\xspace}

\title{Reducing graph transversals via edge contractions} 

\titlerunning{Reducing graph transversals via edge contractions}

\author{Paloma T. Lima}{Department of Informatics, University of Bergen, Norway}{paloma.lima@uib.no}{}{}

\author{Vinicius F.\ dos Santos}{Departamento de Ci\^encia da Computa\c{c}\~{a}o,
	Universidade Federal de Minas Gerais, Belo Horizonte, Brazil}
{viniciussantos@dcc.ufmg.br}{https://orcid.org/0000-0002-4608-4559}
{Grant APQ-02592-16 Minas Gerais Research Support Foundation (FAPEMIG) and Grants 311679/2018-8 and 421660/2016-3 National Council for Scientific and Technological Development (CNPq).}

\author{Ignasi Sau}{LIRMM, Universit\'e de Montpellier, CNRS, Montpellier, France}{ignasi.sau@lirmm.fr}{https://orcid.org/0000-0002-8981-9287}{CAPES-PRINT Institutional Internationalization Program, process 88887.371209/ 2019-00, and projects DEMOGRAPH (ANR-16-CE40-0028), ESIGMA (ANR-17-CE23-0010),  ELIT (ANR-20-CE48-0008-01), and UTMA (ANR-20-CE92-0027).}

\author{U\'everton S. Souza}{Instituto de Computa\c c\~ao, Universidade Federal Fluminense, Niter\'oi, Brazil}{ueverton@ic.uff.br}{https://orcid.org/0000-0002-5320-9209}{Grant E-26/203.272/2017 Rio de Janeiro Research Support Foundation (FAPERJ) and Grant 303726/2017-2 National Council for Scientific and Technological Development (CNPq).}

\authorrunning{P.\ T.\ Lima, V.\ F.\ dos Santos, I.\ Sau and U.\ S.\ Souza}

\Copyright{Paloma T.\ Lima, Vinicius F.\ dos Santos, Ignasi Sau and U{\' e}verton S.\ Souza}

\ccsdesc{Theory of computation~Design and analysis of algorithms}
\ccsdesc{Theory of computation~Graph algorithms analysis}
\ccsdesc{Theory of computation~Parameterized
complexity and exact algorithm}

\keywords{blocker problem, edge contraction, graph transversal, parameterized complexity, vertex cover, feedback vertex set, odd cycle transversal.}

\category{}

\relatedversion{A conference version of this article appeared in the \emph{Proceedings of the 45th International Symposium on Mathematical Foundations of Computer Science (MFCS), volume 170 of LIPIcs, pages 64:1--64:15, \textbf{2020}}. A full version of the paper is permanently available at \url{https://arxiv.org/abs/2005.01460}.}

\supplement{}

\nolinenumbers 

\hideLIPIcs  

\EventEditors{John Q. Open and Joan R. Access}
\EventNoEds{2}
\EventLongTitle{45th International Symposium on Mathematical Foundations of Computer Science (MFCS 2020)}
\EventShortTitle{MFCS 2020}
\EventAcronym{MFCS}
\EventYear{2016}
\EventDate{December 24--27, 2016}
\EventLocation{Little Whinging, United Kingdom}
\EventLogo{}
\SeriesVolume{}
\ArticleNo{1}

\begin{document}

\maketitle

\begin{abstract}
For a graph invariant $\pi$, the \textsc{Contraction($\pi$)} problem consists of, given a graph $G$ and two positive integers $k,d$, deciding whether one can contract at most $k$ edges of $G$ to obtain a graph in which $\pi$ has dropped by at least $d$. Galby et al.~[ISAAC 2019, MFCS 2019] recently studied the case where $\pi$ is the size of a minimum dominating set. We focus on graph invariants defined as the minimum size of a vertex set that hits all the occurrences of graphs in a  collection $\Hcal$ according to a fixed containment relation. We prove \coNP-hardness results under some assumptions on the graphs in $\Hcal$, which in particular imply that \textsc{Contraction($\pi$)} is \coNP-hard even for fixed $k=d=1$ when $\pi$ is the size of a minimum feedback vertex set or an odd cycle transversal. In sharp contrast, we show that when $\pi$ is the size of a minimum vertex cover, the problem is in \XP parameterized by $d$.
\end{abstract}

\newpage

\section{Introduction}
\label{sec:intro}


Graph modification problems play a central role in algorithmic graph theory and have been widely studied in the last few years~\cite{crespelle2020survey,FominSM15grap,BodlaenderHL14grap}. In this kind of problem, given a graph, we want to perform a small number of modifications so that the resulting graph satisfies a desired property. Typically, this property is described as a graph class to which the resulting graph must belong, and the corresponding problem is usually \NP-hard~\cite{10.1145/800133.804355,WatanabeAN83,LewisY80}. Numerous famous problems can be stated as graph modification problems. For instance, if the operation is vertex deletion and the target graph class is that of forests, we obtain the well-known {\sc Feedback Vertex Set} problem. A distinct type of graph modification problem that has been considered more recently is concerned with graph invariants\footnote{We use ``invariant'' instead of parameter in order to avoid confusion with the parameter of the corresponding parameterized problem.}, instead of graph classes. The goal here is to perform a small number of modifications in order to decrease (or increase) a given invariant of the input graph. These are the so-called \emph{blocker problems}, the main object of study in this work.

More precisely, in a blocker problem with invariant $\pi$, given a graph $G$ and a set $\mathcal{M}$ of graph modification operations, the question is whether $G$ can be modified into a graph $G'$ such that $\pi(G')\leq \pi(G)-d$, for some \emph{threshold} $d$, via at most $k$ operations from $\mathcal{M}$. The name {\sl blocker} comes from the fact that the set of vertices or edges involved in the modifications can be viewed as ``blocking'' the invariant $\pi$, that is, preventing $\pi$ from being smaller, as we would like in a minimization problem. Identifying parts of the graph that are responsible for an increase in a graph invariant gives useful information about the graph structure and has been the central question around many graph problems. For instance, if the invariant in question is the size of a longest path, $d=1$ and the operation is vertex deletion, the problem becomes equivalent to testing whether there exists a set of $k$ vertices that intersects every longest path of the input graph~\cite{RautenbachS14,CHEN2017287,CERIOLI2020DM,CERIOLI2019}. Another example is that of computing the Hadwiger number of a graph. The {\sc Hadwiger Number} problem takes as input a graph $G$ and an integer $t$, and asks whether there exists a set of edges in $G$ the contraction of which results in a graph isomorphic to the complete graph on $t$ vertices~\cite{BOLLOBAS1980195,doi:10.1137/140975279,fomin2020computation}. This problem can be formulated as a blocker problem with the edge contraction operation, the invariant being the independence number (denoted by $\alpha$), $k=|V(G)|-t$, and $d=\alpha(G)-1$.

Because of their relevance and connection to other well-studied graph problems, blocker problems have been investigated for numerous graph invariants, such as the chromatic number, the independence number, the matching number, the diameter, the domination number, the total domination number, and the clique number of a graph~\cite{BTT11,CWP11, bazgan2015blockers,diner2018contraction,kimarxiv,PBP,paulusma2018critical,RBPDCZ10,WatanabeAN83,GalbyArxivTotalD,GalbyArxivP3}. The set $\mathcal{M}$ has so far been restricted to contain a single operation, usually vertex deletion, edge deletion, edge addition, or edge contraction. In this work, we restrict ourselves to the edge contraction operation. Formally, we are interested in the following  problem, where $\pi$ is any graph invariant.

\problemdef
	{\contracpi}
	{A graph $G$ and two positive integers $k,d$.}
	{Can $G$ be $k$-contracted into a graph $G'$ such that $\pi(G')\leq \pi(G)-d$?}

When $k$ and $d$ are fixed instead of being part of the input, we denote the corresponding problem by \contracpidk. Blocker problems with the edge contraction operation have already been studied with respect to the chromatic number, clique number, and independence number~\cite{diner2018contraction,paulusma2018critical}, and the domination number~\cite{GALBY2021DM}, denoted by $\chi$, $\omega$, $\alpha$, and $\gamma$, respectively. These works address the problem from the point of view of graph classes. Diner et al.~\cite{diner2018contraction} showed, among other results, that \contracpi is \NP-complete restricted to split graphs for $\pi\in\{\chi,\alpha,\omega\}$, but it is polynomial-time solvable in this graph class for fixed~$d$ in all three cases. Galby et al.~\cite{GALBY2021DM} recently initiated the study of the problem for $\pi=\gamma$ for the case $d=1$,
 providing several negative and positive results restricted to particular graph classes, such as a polynomial-time algorithm for {\sc $k$-Contraction$(\gamma,1)$} on $(P_5+pK_1)$-free graphs, for any $p\geq 1$.
 Galby et al.~\cite{GALBY2021DM} also considered a variant of the blocker problem in which an edge is \emph{given as part of the input}. Namely, they showed that the problem of deciding whether the contraction of this specific edge decreases the domination number of a graph admits no polynomial-time algorithm unless $\P$=$\NP$. We observe here that their proof~\cite[Theorem 3.13]{GALBY2021DM} in fact works for any graph invariant satisfying two specific conditions, as stated in the following proposition.
	
\begin{proposition}[Galby et al.~\cite{GALBY2021DM}]
\label{prop:no-poly-fixed-edge}
Let $\pi$ be a graph invariant such that
\begin{enumerate}
\item[(i)] it is \NP-hard to compute the $\pi$-number of a graph and
\item[(ii)] contracting an edge reduces $\pi$ by at most one.
\end{enumerate}
Then, there exists no polynomial-time algorithm deciding whether contracting one given edge decreases the $\pi$-number of a graph, unless $\P$=$\NP$.
\end{proposition}


In this work, the invariants we focus on are \emph{$\mathcal{H}$-transversals}, that is, the minimum size of a vertex set of a graph that hits all the occurrences of graphs in a fixed (finite or infinite) collection $\mathcal{H}$  according to a specified containment relation $\prec$. We denote this invariant by $\tau_\mathcal{H}^\prec$. Note that distinct instantiations of $\mathcal{H}$ and $\prec$ capture, for instance, the vertex cover, feedback vertex set, and  odd cycle transversal numbers, and that these three invariants satisfy the conditions of Proposition~\ref{prop:no-poly-fixed-edge}.

\medskip
\noindent \textbf{Our results and techniques}. We show (Theorem~\ref{thm:hard-Hcal}) that \contracHcalone is \coNP-hard when $\mathcal{H}$ is a family of 2-connected graphs containing at least one non-complete graph and $\prec$ is any of the subgraph, induced subgraph, minor, or topological minor containment relations. This implies that it is \coNP-hard to test whether we can reduce the feedback vertex set number or the odd cycle transversal number of a graph by performing {\sl one} edge contraction. Note that this result is not implied by  Proposition~\ref{prop:no-poly-fixed-edge}, since we do {\sl not} specify which edge should be contracted. We also show (Theorem~\ref{thm:hard-cliques}) that the problem is \coNP-hard if $\mathcal{H}$ is a family of cliques of size at least three and $\prec$ is the minor or topological minor containment relation. The same holds (Theorem~\ref{thm:hard-paths}) if $\mathcal{H}$ is a family of graphs containing a path on at least four vertices and any collection of 2-connected graphs and $\prec$ is the subgraph, induced subgraph, minor, or topological minor containment relation. All these reductions are from the 3-\textsc{Sat} problem restricted to {\sl clean} formulas (see Section~\ref{sec:prelim} for the definition).

\smallskip

We point out that, as can be seen by earlier results and the ones mentioned  above, blocker problems are generally very hard, and become polynomial-time solvable only when restricted to specific graph classes. However, we show that the picture changes completely when the invariant in question is the vertex cover number of a graph (denoted by $\vc$): we prove (Theorem~\ref{thm:vc-XP}) that \contracvc can be solved in \XP time parameterized by~$d$ on {\sl general graphs}, hence in polynomial time for fixed $d$, in particular for $d=1$. This result should be compared to Proposition~\ref{prop:no-poly-fixed-edge}, which shows that the problem is hard for $d=1$ if the edge to be contracted is {\sl prescribed}. Note that since the contraction of an edge may drop the minimum vertex cover of a graph by at most one, we may assume that $k \geq d$ (as otherwise the answer is trivially `\no'), hence parameter $d$ is {\sl stronger} than $k$.  Our algorithm (cf. Algorithm~\ref{alg:XPvc}) starts by checking whether the \emph{bipartite contraction number} of $G$ (i.e., the minimum number of edges to be contracted in order to obtain a bipartite graph), denoted by $\bc(G)$, is at most $d-1$. This can be done in \FPT time by a result of Heggernes et al.~\cite{HeggernesHLP13}, later improved by Guillemot and Marx~\cite{GUILLEMOTMARX}.
If $\bc(G) \geq d$, a simple argument allows to conclude that we are dealing with a \yes-instance. Otherwise, we distinguish two cases depending on whether $G$ contains a connected component $C$ with $\vc(C) > d$ or not. If it is {\sl not} the case, we show (Lemma~\ref{lem:smallVC}) that the problem can be solved in \FPT time by combining a formulation in \MSO logic and a dynamic programming algorithm. Otherwise, we prove that we may assume (Lemma~\ref{lem:2approxVC}) that $k< 2d$, which enables us to enumerate all subsets $F \subseteq E(G)$ of size at most $k$ and, for each of them, solve the problem in \FPT time by a branching algorithm,
exploiting the fact that $\vc$ can be computed in polynomial time on bipartite graphs.

Finally, we also show that a small modification of the above algorithm yields (Corollary~\ref{cor:vc-2-approx}) that the problem of determining the minimum number of edges to be contracted to drop the vertex cover number of a graph by $d$ can be 2-approximated in \FPT-time parameterized by $d$.


\medskip
\noindent \textbf{Organization}. In Section~\ref{sec:prelim} we provide some preliminaries and formally define all the problems mentioned throughout the text. In Section~\ref{sec:hard} we prove the \coNP-hardness results, and in Section~\ref{sec:vertex-cover} we present the algorithms for reducing the size of a minimum vertex cover via edge contractions. We conclude the article in Section~\ref{sec:concl} with some further observations and  directions for further research.





	


\section{Preliminaries}
\label{sec:prelim}

\noindent\textbf{Graph notation.} We use standard graph-theoretic notation, and we refer the reader to~\cite{Diestel12} for any undefined notation. We will only consider undirected graphs without loops nor multiple edges, and we denote an edge between two vertices $u$ and $v$ by $\{u,v\}$. A subgraph $H$ of a graph $G$ is \emph{induced} if $H$ can be obtained from $G$ by deleting vertices. A graph $G$ is \emph{$H$-free} if it does not contain any induced subgraph isomorphic to $H$.  For a graph $G$ and a set $S \subseteq V(G)$, we use the notation $G \setminus S :=G[V(G) \setminus S]$. For a set $F \subseteq E(G)$, we denote by $V(F)$ the set of vertices incident to some edge in $F$. We denote by $\Delta(G)$ (resp. $\omega(G)$) the maximum vertex degree (resp. clique size) of a graph $G$.
For an integer $h \geq 1$, we denote by $P_h$ (resp. $C_h$, $K_h$) the path (resp. cycle, clique) on $h$ vertices. A \emph{star} is a tree with at least one edge in which one vertex is adjacent to all other vertices.
A vertex set $S$ of a connected graph $G$ is a \emph{separator} if $G \setminus S$ is disconnected. 
For an integer $k\geq 1$, a graph $G$ is \emph{$k$-connected} if it is connected and does not have any separator of size at most $k-1$. For an integer $k \geq 1$, we denote by $[k]$ the set of all integers $i$ such that $1 \leq i \leq k$.

The \emph{open} (resp. \emph{closed}) \emph{neighborhood} of a vertex $v$ in a graph $G$ is denoted by $N_G(v)$ (resp. $N_G[v]$). We may drop the subscript if the graph $G$ is clear from the context. The \emph{contraction} of an edge $e = \{u,v\}$ in a graph $G$ results in a graph $G'$ obtained from $G$ by removing $u$ and $v$, and adding a new vertex $v_e$ with $N_{G'}(v_e) = N_G(u) \cup N_G(v)$. We denote by $G/e$ the graph obtained from $G$ by contracting an edge $e$, and if $F \subseteq E(G)$, we denote by $G/F$ the graph obtained from $G$ by contracting all the edges in $F$, in any order; it is easy to verify that the resulting graph does not depend on the order in which the contractions are applied. If $|F| = k$, we say that $G$ is \emph{$k$-contracted} into $G / F$.

A \emph{vertex cover} (resp. \emph{feedback vertex set}, \emph{odd cycle transversal}) of a graph $G$ is a set $S \subseteq V(G)$ such that $G \setminus S$ is edgeless (resp. acyclic, bipartite). We denote the minimum size of a vertex cover (resp. feedback vertex set, odd cycle transversal) of a graph $G$ by $\vc(G)$ (resp. $\fvs(G)$, $\oct(G)$). Note that, if $G$ is connected, $\vc(G) = 1$ if and only if $G$ is a star.

A graph $H$ is a \emph{minor} of a graph $G$ if it can be obtained from $G$ by removing vertices, deleting edges, and contracting edges. A graph $H$ is a \emph{topological minor} of a graph $G$ if it can be obtained from $G$ by removing vertices, deleting edges, and contracting edges having at least one vertex of degree at most two. The operation of \emph{subdividing} an edge $\{u,v\}$ consists in deleting the edge $\{u,v\}$ and adding a new vertex $w$ adjacent to $u$ and $v$.

\medskip
\noindent\textbf{Graph transversals.} For a fixed graph containment relation $\prec$ and a fixed (finite or infinite) collection of graphs $\Hcal$, we define the invariant $\tau_{\Hcal}^{\prec}$ such that, for every graph $G$, $\tau_{\Hcal}^{\prec}(G)$ is equal to the minimum size of a set
$S \subseteq V(G)$ such that $G \setminus S$ does not contain any of the graphs in $\Hcal$ according to containment relation $\prec$. If $\Hcal = \{H\}$, we denote $\tau_{\{H\}}^{\prec}$ by $\tau_{H}^{\prec}$. Such a set $S$ is called an $\Hcal$-transversal or an $\Hcal$-hitting set.

For instance, if $\prec$ is the minor relation and $H$ is an edge (resp. a triangle), then $\tau_{H}^{\prec}$ is the size of a minimum vertex cover (resp. feedback vertex set), which we abbreviate as  $\vc$ (resp. $\fvs$). On the other hand, if $\prec$ is the subgraph relation and $\Hcal$ contains all odd cycles, then $\tau_{\Hcal}^{\prec}$ is the size of a minimum odd cycle transversal, which we abbreviate as $\oct$.

%

\medskip
\noindent\textbf{Definition of the problems.} For a graph invariant $\pi$,
%
%
%
%
we also consider the versions of the $\contracpi$ problem, defined in the Introduction, where one or both positive integers $k$ and $d$ are {\sl fixed}, instead of being part of the input. Namely, we denote by \contracpik, \contracpid, and \contracpidk the version of $\contracpi$ in which $k$, $d$, and both $k$ and $d$ are fixed, respectively.

%
%

Finally, we define the following optimization version of \contracpi.

\problemOPTdef
	{\mincontracpi}
	{A graph $G$ and a positive integer $d$.}
	{The minimum integer $k$ such that $G$ be $k$-contracted into a graph $G'$ such that $\pi(G')\leq \pi(G)-d$.}

\noindent\textbf{Parameterized complexity.} We refer the reader to~\cite{DF13,CyganFKLMPPS15} for basic background on parameterized complexity, and we recall here only some basic definitions.
A \emph{parameterized problem} is a decision problem whose instances are pairs $(x,k) \in \Sigma^* \times \mathbb{N}$, where $k$ is called the \emph{parameter}.
A parameterized problem is \emph{fixed-parameter tractable} (\FPT) if there exists an algorithm $\Acal$, a computable function $f$, and a constant $c$ such that given an instance $I=(x,k)$,
$\Acal$ (called an {\sf FPT} \emph{algorithm}) correctly decides whether $I \in L$ in time bounded by $f(k) \cdot |I|^c$. A parameterized problem is \emph{slice-wise polynomial} ({\sf XP}) if there exists an algorithm $\Acal$ and two computable functions $f,g$ such that given an instance $I=(x,k)$,
$\Acal$ (called an {\sf XP} \emph{algorithm}) correctly decides whether $I \in L$ in time bounded by $f(k) \cdot |I|^{g(k)}$.


%

Within parameterized problems, the class {\sf W}[1] may be seen as the parameterized equivalent to the class \NP of classical optimization problems. Without entering into details (see~\cite{DF13,CyganFKLMPPS15} for the formal definitions), a parameterized problem being {\sf W}[1]-\emph{hard} can be seen as a strong evidence that this problem is {\sl not} \FPT. The canonical example of {\sf W}[1]-hard problem is \textsc{Independent Set} parameterized by the size of the solution.
To transfer ${\sf W}[1]$-hardness from one problem to another, one uses a \emph{parameterized reduction}, which given an input $I=(x,k)$ of the source problem, computes in time $f(k) \cdot |I|^c$, for some computable function $f$ and a constant $c$, an equivalent instance $I'=(x',k')$ of the target problem, such that $k'$ is bounded by a function depending only on $k$. An equivalent definition of $\W$[1]-hard problem is any problem that admits a parameterized reduction from \textsc{Independent Set} parameterized by the size of the solution.

\medskip
\noindent
\textbf{Treewidth and Courcelle's Theorem.} For an integer $k \geq 1$, a \emph{$k$-tree} is a graph that be obtained from a $k$-clique by recursively adding vertices adjacent to a $k$-clique of the current graph. The \emph{treewidth} of a graph $G$, denoted by $\tw(G)$, is the smallest integer $k$ such that $G$ is a subgraph of a $k$-tree. The syntax of \emph{monadic second order} (\MSO) logic of graphs includes the logical connectives $\vee$, $\wedge$, $\neg$, variables for vertices, edges, sets of vertices and sets of edges, the quantifiers $\forall, \exists$ that can be applied to these variables, and the binary relations expressing whether a vertex or an edge belong to a set, whether an edge is incident to vertex, whether two vertices are adjacent, and whether two sets are equal. The following result of Courcelle~\cite{Courcelle90}, as well as one of its several optimization variants~\cite{ArnborgLS91}, is one of the most widely used results in the area of parameterized complexity.

\begin{proposition}[Courcelle~\cite{Courcelle90}, Arnborg et al.~\cite{ArnborgLS91}]
\label{thm:Courcelle}
Checking whether an \MSO formula $\varphi$ holds on an $n$-vertex graph of treewidth at most $\tw$ can be done in time $f(\varphi, \tw) \cdot n$, for a computable function $f$. Moreover, within the same running time, one can find a vertex or edge set of $G$ of maximum or minimum size that satisfies $\varphi$.
\end{proposition}

\noindent
\textbf{Exponential Time Hypothesis and clean} 3-\textsc{Sat}. The \emph{Exponential Time Hypothesis} (\ETH) of Impagliazzo and Paturi~\cite{ImpagliazzoP01-ETH} implies that the 3-\textsc{Sat} problem on $n$ variables cannot be solved in time $2^{o(n)}$.
The Sparsification Lemma of Impagliazzo et al.~\cite{ImpagliazzoP01} implies that if the \ETH holds, then there is no algorithm solving a 3-\textsc{Sat} formula with $n$ variables and $m$ clauses in time $2^{o(n+m)}$. Using the terminology from Cygan et al.~\cite{CyganMPP17}, a 3-\textsc{Sat} formula $\varphi$, in conjunctive normal form, is said to be \emph{clean} if each variable of $\varphi$ appears exactly three times, at least once positively and at least once negatively, and each clause of $\varphi$ contains two or three literals and does not contain twice the same variable. Cygan et al.~\cite{CyganMPP17} observed the following useful lemma.

\begin{lemma}[Cygan et al.~\cite{CyganMPP17}]\label{lem:clean3SAT}
The problem of deciding whether a clean 3-\textsc{Sat} formula with $n$ variables is satisfiable is \NP-hard, and the existence of an algorithm in time $2^{o(n)}$ to solve it would violate the \ETH.
\end{lemma}

\section{Hardness results}
\label{sec:hard}

%

We start with some definitions that will be used in the reductions of this section. Let $G$ and $H$ be two graphs, let $u,v \in V(H)$, and let $\{x,y\} \in E(G)$. By \emph{replacing  $\{x,y\}$ by $H_{u,v}$} we mean deleting edge $\{x,y\}$ from $G$, adding a copy of $H$ and identifying vertices $u$ and $v$ of  $H$ with vertices $x$ and $y$ of $G$, respectively. The operation of \emph{replacing  $\{x,y\}$ by two copies of $H_{u,v}$} is defined similarly, except that we add two copies of $H$ and we identify vertices $u$ and $v$ of both copies of $H$ with vertices $x$ and $y$ of $G$, respectively. By \emph{attaching $H_u$ to} $x \in V(G)$ we mean adding a copy of $H$ and identifying vertex $u$ of $H$ with vertex $x$ of $G$, and by \emph{attaching a pendent $H_u$ to} $x \in V(G)$ we mean adding a copy of $H$ and an edge between vertex $u$ of $H$ and vertex $x$ of $G$. We denote by $H^2_u$  the graph obtained from two copies of $H$ by identifying vertex $u$ in each of the copies.

\begin{theorem}\label{thm:hard-Hcal}
Let $\Hcal$ be a collection of 2-connected graphs containing at least one non-complete graph. Then  \contracHcalone is \coNP-hard, for $\prec$ being any of the subgraph, induced subgraph, minor, or topological minor containment relations. Moreover, the problem cannot be solved in subexponential time assuming the \ETH, even restricted to graphs with maximum degree depending on $\Hcal$.
\end{theorem}
\begin{proof}
We present a reduction from the 3-\textsc{Sat} problem restricted to clean formulas, which is \NP-hard by Lemma~\ref{lem:clean3SAT}. Namely, given a clean formula $\varphi$ with $n$ variables and $m$ clauses, we will construct in polynomial time an instance $G_{\varphi}^{\Hcal}$ such that $\varphi$ is satisfiable if and only if $G_{\varphi}^{\Hcal}$ is a \no-instance of \contracHcalone.
We start by constructing a graph $G_{\varphi}$ that will be reused in the other reductions of this section, and which is inspired by the classical \NP-hardness reduction~\cite{GareyJ79} from 3-\textsc{Sat} to \textsc{Vertex Cover}.

For each variable $x$ of $\varphi$ and for each clause $C$ containing $x$ in a literal $\ell \in \{x, \bar{x}\}$, we add to $G_{\varphi}$ a new vertex $a_{x,C, \ell}$. We also introduce another ``dummy'' vertex $a_x$. Since $\varphi$ is clean, we have introduced four vertices in  $G_\varphi$ for each variable $x$. Let $a_{x,C_1, \ell}$, $a_{x,C_2, \bar{\ell}}$, $a_{x,C_3, \ell}$, $a_{x}$ be the four introduced vertices (recall that $x$ appears at least once positively and negatively in $\varphi$). We add the following four edges, inducing a $C_4$: $(a_{x,C_1, \ell}, a_{x,C_2, \bar{\ell}})$, $(a_{x,C_2, \bar{\ell}}, a_{x,C_3, \ell})$, $(a_{x,C_3, \ell}, a_x)$, and $(a_x, a_{x,C_1, \ell})$. We denote by $A$ the union of all the vertices in these variable gadgets.

For each clause $C$ of $\varphi$ and for each literal $\ell$ in $C$, we add to $G_{\varphi}$ a new vertex $b_{C, \ell}$. Since $\varphi$ is clean, we have introduced two or three vertices in $G_{\varphi}$ for each clause $C$. We add an edge  between every pair of these vertices, hence inducing a clique of size two or three. We denote by $B$ the union of all the vertices in these clause gadgets.

Finally, for each variable $x$ of $\varphi$ and for each clause $C$ containing $x$ in a literal $\ell \in \{x, \bar{x}\}$, we add to $G_{\varphi}$ an edge between $a_{x,C, \ell} \in A$ and $b_{C, \ell} \in B$.
This concludes the construction of $G_{\varphi}$, which we proceed to modify. Note that $V(G_{\varphi})= A \cup B$.

Let $H \in \Hcal $ be a non-complete 2-connected graph, and let $u,v$ be two non-adjacent vertices in $H$. Starting from $G_{\varphi}$, we replace each of the edges between two vertices in $A$ or two vertices in $B$ by two copies of $H_{u,v}$, and each edge between a vertex in $A$ and a vertex in $B$ by one copy of $H_{u,v}$. Each of these copies of $H$ is called an \emph{$A$-copy}, \emph{$B$-copy}, or \emph{$AB$-copy}, depending on whether its attachment vertices are both in $A$, both in $B$, or one in $A$ and one in $B$, respectively.

Finally, for each $AB$-copy of $H$, we choose arbitrarily within it a vertex $z$ distinct from $u$ and $v$, and we attach a pendent copy of $H^2_u$ to $z$. These newly added copies of $H$ are called \emph{pendent copies}, the edge linking them to its corresponding $AB$-copy of $H$ is called the \emph{pendent edge} of that $AB$-copy of $H$, and the vertex in the $AB$-copy incident with the pendent edge is called the \emph{base vertex} of that $AB$-copy of $H$. This concludes the construction of $G_{\varphi}^{\Hcal}$; see Figure~\ref{fig:exampleC4} for an example for $\Hcal$ containing all cycles and $H = C_4$.

\begin{figure}[ht]
\begin{center}
\includegraphics[scale=1.12]{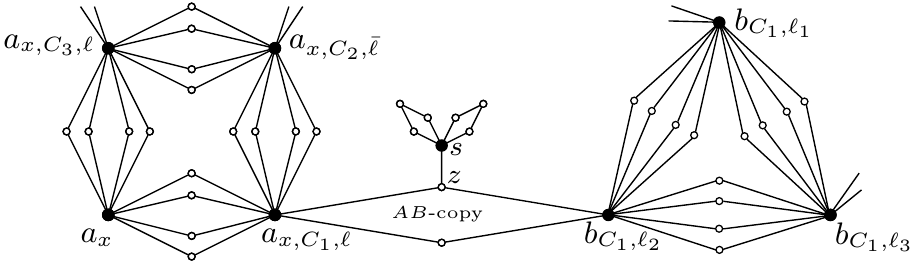}
\end{center}
\caption{Illustration of the graph $G_{\varphi}^{\Hcal}$ for $\Hcal$ containing all cycles and $H = C_4$. Black (resp. white) vertices are attachment (resp. internal) vertices of the corresponding copies of $H$. Vertex $z$ is the base vertex and $\{z,s\}$ is the pendent edge of the depicted $AB$-copy of $H$.}
\label{fig:exampleC4}
\end{figure}

We can clearly assume that $\Hcal$ is an antichain with respect to $\prec$, that is, that its elements are pairwise incomparable with respect to $\prec$.  Assume for the sake of presentation that $\prec$ is the subgraph relation, and we omit it from the notation $\tau_{\Hcal}^{\prec}$; at the end of the proof we will argue that the same arguments apply as well to the other containment relations listed in the statement of the theorem. Note that $G_{\varphi}^{\Hcal}$ contains $2n$ pairwise vertex-disjoint $A$-copies of $H$ and $3n$ pairwise vertex-disjoint pendent copies of $H$, taking into account that each variable appears exactly three times in $\varphi$. On the other hand, since for clause $C$ of $\varphi$ the vertices $\{ b_{C, \ell} \mid \ell \in C\}\subseteq B$ induce a clique in $G_{\varphi}$, for each clause $C$ of $\varphi$ at least $|C|-1$ vertices of $G_{\varphi}^{\Hcal}$ are needed to hit the $B$-copies of $H$ among the vertices  $\{ b_{C, \ell} \mid \ell \in C\}\subseteq B$, where $|C|$ denotes the number of literals in $C$. Therefore, using $\varphi$ that is clean, it follows that
\begin{equation}\label{eq:budget}
\tau_{\Hcal}(G_{\varphi}^{\Hcal})\ \geq\ 2n +  3n + \sum_{C \in \varphi}(|C|-1)\ =\ 8n - m.
\end{equation}
We present three claims that, together, will conclude the proof of the theorem.

\begin{claim}\label{claim:likeVC}
$\tau_{\Hcal}(G_{\varphi}^{\Hcal}) = 8n - m$ if and only if $\varphi$ is satisfiable.
\end{claim}
\begin{proof}
Suppose first that $\varphi$ is satisfiable, and let $\alpha$ be an assignment of the variables that satisfies all the clauses in $\varphi$. We define a set $X \subseteq V(G_{\varphi}^{\Hcal})$ as follows. For each variable $x$, add to $X$ all vertices $a_{x,C,\ell}$ such that $\alpha(\ell)$ is {\sf true}. If only one vertex was added in the previous step, add to $X$ vertex $a_x$ as well. For each clause $C$, choose a literal $\ell$ that satisfies $C$, and add to $X$ the set $\{b_{C,\ell'} \mid \ell' \neq \ell\}$. Finally, for each pendent copy of $H$, add to $X$ its corresponding vertex $u$. By construction we have that $|X|=8n - m$, hence by Equation~(\ref{eq:budget}) we just have to verify that $G_{\varphi}^{\Hcal} \setminus X$ does not contain any of the graphs in $\Hcal$ as a subgraph. Note that by the choice of $X$, it contains at least one vertex of each $A$-copy, $B$-copy, and pendent copy of $H$, and since $\alpha$ is a satisfying assignment, $X$ contains at least one vertex of each $AB$-copy of $H$ as well.
Let $F$ be a connected component of $G_{\varphi}^{\Hcal} \setminus X$. Since by hypothesis all the graphs in $\Hcal$ are $2$-connected, it suffices to verify that no $2$-connected component $F'$ of $F$ contains one of the graphs in $\Hcal$ as a subgraph. By the construction of $G_{\varphi}^{\Hcal}$, the choice of $X$, and the fact that $X$ hits all copies of $H$ in $G_{\varphi}^{\Hcal}$, it follows that such a $2$-connected component $F'$ is a {\sl proper} subgraph of $H$. Since $\Hcal$ is an antichain, $F'$ cannot contain any of the graphs in $\Hcal$ as a subgraph, and we are done.

\smallskip

Conversely, suppose that there exists $X \subseteq V(G_{\varphi}^{\Hcal})$ with $|X| \leq 8n-m$ such that $G_{\varphi}^{\Hcal} \setminus X$ does not contain any of the graphs in $\Hcal$, in particular $H$, as a subgraph. By Equation~(\ref{eq:budget}), we have that $|X| = 8n-m$.
By construction of $G_{\varphi}^{\Hcal}$ and the fact that $|X| = 8n-m$ (see the paragraph above Equation~(\ref{eq:budget})), it follows that $X$ must contain {\sl exactly one} of the pairs $\{ a_{x,C_1,\ell}, a_{x,C_3,\ell} \}$ and $\{ a_x,a_{x,C_2,\bar{\ell}} \}$ for each variable $x$, and {\sl exactly} $|C|-1$ vertices in $\{b_{C,\ell} \mid \ell \in C\}$ for each clause $C$. We define the following assignment $\alpha$ of the variables: for each variable $x$, let $\ell \in \{x, \bar{x}\}$ such that $a_{x,C,\ell} \in X$ for some clause $C$. Then we set $\alpha(x)$ to {\sf true} if $\ell = x$, and to {\sf false} if $\ell = \bar{x}$. By the above discussion, this is a valid assignment. Consider a clause $C$ of $\varphi$, and let $\ell$ be the literal in $C$ such that $b_{C,\ell} \notin X$. Since  $G_{\varphi}^{\Hcal} \setminus X$  does not contain $H$ as a subgraph, there must exist a variable $x \in \{\ell, \bar{\ell}\}$ such that $a_{x,C,\ell} \in X$, as otherwise the $AB$-copy of $H$ between $b_{C,\ell}$ and $a_{x,C,\ell}$
 would be an occurrence of $H$ in $G_{\varphi}^{\Hcal} \setminus X$. By the definition of $\alpha$, necessarily $\alpha(\ell)$ is {\sf true}, and therefore $\alpha$ satisfies $C$. Since this argument holds for every clause, we conclude that $\varphi$ is satisfiable.
\end{proof}

\begin{claim}\label{claim:no-edge}
If $\tau_{\Hcal}(G_{\varphi}^{\Hcal}) = 8n - m$, then there is no edge $e$ such that $\tau_{\Hcal}(G_{\varphi}^{\Hcal} / e) < \tau_{\Hcal}(G_{\varphi}^{\Hcal})$.
\end{claim}
\begin{proof}
Let $e \in E(G_{\varphi}^{\Hcal})$ be an arbitrary edge, and consider the graph $G_{\varphi}^{\Hcal} / e$. Since in $G_{\varphi}^{\Hcal}$ there are {\sl two} copies of $H$ between every pair of vertices in $A$ and $B$ within the same variable or clause gadget, respectively, {\sl two} copies of $H$ are attached to $u$ in every pendent $H_u^2$, and $u$ and $v$ are {\sl not} adjacent in $H$, it follows that, for each such a pair of copies of $H$,  at least one of them still survives in $G_{\varphi}^{\Hcal} / e$.

Therefore, $G_{\varphi}^{\Hcal} / e$ still contains $2n$ pairwise vertex-disjoint $A$-copies of $H$, $3n$ pairwise vertex-disjoint pendent copies of $H$, and, for each clause $C$ of $\varphi$, at least $|C|-1$ vertices of $G_{\varphi}^{\Hcal} / e$ are needed to hit the $B$-copies of $H$ among the vertices  $\{ b_{C, \ell} \mid \ell \in C\}\subseteq B$. Thus, as by hypothesis $\tau_{\Hcal}(G_{\varphi}^{\Hcal}) = 8n - m$, we have that
$$
\tau_{\Hcal}(G_{\varphi}^{\Hcal} / e)\ \geq\ 2n +  3n + \sum_{C \in \varphi}(|C|-1)\ =\ 8n - m \ = \ \tau_{\Hcal}(G_{\varphi}^{\Hcal}),
$$
and the claim follows.
\end{proof}

\begin{claim}\label{claim:edge}
If $\tau_{\Hcal}(G_{\varphi}^{\Hcal}) > 8n - m$, then there is an edge $e$ such that $\tau_{\Hcal}(G_{\varphi}^{\Hcal} / e) < \tau_{\Hcal}(G_{\varphi}^{\Hcal})$.
\end{claim}
\begin{proof}
Let $X \subseteq V(G_{\varphi}^{\Hcal})$ be an $\Hcal$-hitting set of minimum size. We call a vertex in a copy of $H$ in $G_{\varphi}^{\Hcal}$ \emph{internal} if it is distinct from its attachment vertices (cf. Figure~\ref{fig:exampleC4}). We proceed to construct another $\Hcal$-hitting set $X'  \subseteq V(G_{\varphi}^{\Hcal})$ with canonical properties, such that
\begin{enumerate}
\item\label{item:1} $|X'| \leq |X|$,
\item\label{item:2}  $X'$ contains exactly either $ \{a_{x,C_1,\ell}, a_{x,C_3,\ell} \}$ or $\{ a_x,a_{x,C_2,\bar{\ell}} \}$ for each variable $x$,
\item\label{item:3} $X'$ contains exactly $|C|-1$ vertices in the set $\{b_{C,\ell} \mid \ell \in C\}$ for each clause $C$,
\item\label{item:4} $X'$ contains exactly one vertex in each pair of pendent copies of $H$,
\item\label{item:5} $X'$ contains no internal vertex of an $A$-copy, $B$-copy, or pendent copy of $H$, and
\item\label{item:6} all internal vertices of $AB$-copies of $H$ that are in $X'$ are base vertices.
\end{enumerate}
Note that Property~\ref{item:1} above implies that $|X'| = |X| = \tau_{\Hcal}(G_{\varphi}^{\Hcal}) > 8n - m$.

We construct the set $X'$ via the following procedure:
\begin{enumerate}
\item\label{algoX:1} Start with $X' = X$.
\item\label{algoX:2} For each $A$-copy, $B$-copy, or pendent copy $\tilde{H}$ of $H$ in $G_{\varphi}^{\Hcal}$ such that $X'$ contains at least one internal vertex in $\tilde{H}$, remove from $X'$ all internal vertices of $\tilde{H}$, and add to $X'$ any of the attachment vertices of $\tilde{H}$, which may already be in $X'$.
\item\label{algoX:3} For each variable $x$, let $X'_x = X' \cap \{a_{x,C_1,\ell}, a_{x,C_2,\bar{\ell}}, a_{x,C_3,\ell}, a_x\}$. If $|X'_x| \geq 3$, let $P$ be one of the pairs $\{ a_{x,C_1,\ell}, a_{x,C_3,\ell} \}$ and $\{ a_x,a_{x,C_2,\bar{\ell}} \}$ such that $P \subseteq X'$. Remove $X'_x \setminus P$ from $X'$ and, for every vertex $v \in X'_x \setminus P$ with $v \neq a_x$, add to $X'$ an arbitrarily chosen internal vertex in the $AB$-copy of $H$ containing $v$, which may already be in $X'$.
\item\label{algoX:4} For each clause $C$, let $X'_C = X' \cap \{ b_{C, \ell} \mid \ell \in C\}$. Note that by construction of $G_{\varphi}^{\Hcal}$ and Step~\ref{algoX:2} above, $|X'_C| \geq |C| -1$. If $|X'_C| = |C|$, remove from $X'$ an arbitrarily chosen vertex $v \in \{ b_{C, \ell} \mid \ell \in C\}$, and add to $X'$ an arbitrarily chosen internal vertex in the $AB$-copy of $H$ containing $v$, which may already be in $X'$.
\item\label{algoX:5} For each $AB$-copy $\tilde{H}$ of $H$ in $G_{\varphi}^{\Hcal}$ such that $X'$ contains at least one internal vertex in $\tilde{H}$, remove  from $X'$ all internal vertices of $\tilde{H}$, and add to $X'$ the base vertex of $\tilde{H}$, which may already be in $X'$.
\end{enumerate}
\noindent Let $X'$ be the set obtained at the end of the above procedure. It can be easily verified that  $X'$ satisfies the desired Properties~\ref{item:1}-\ref{item:6}. In order to see that $X'$ is a $\Hcal$-hitting set, note that, by construction of $G_{\varphi}^{\Hcal}$, each vertex in $A \cup B$ is contained in at most one $AB$-copy of $H$. Thus, in Steps~\ref{algoX:3} and~\ref{algoX:4} of the above procedure, when we swap vertices in $A \cup B$ by internal vertices in $AB$-copies of $H$, we guarantee that the currently constructed set $X'$ is still a $\Hcal$-hitting set. Clearly, this property is also preserved in Steps~\ref{algoX:2} and~\ref{algoX:5}.

\medskip

We now proceed, using the constructed $\Hcal$-hitting set $X'$, to identify an edge $e^{\star} \in E(G_{\varphi}^{\Hcal})$ such that $\tau_{\Hcal}(G_{\varphi}^{\Hcal} / e^{\star}) < |X'| = \tau_{\Hcal}(G_{\varphi}^{\Hcal})$, concluding the proof of the claim. Since by hypothesis $|X'| = \tau_{\Hcal}(G_{\varphi}^{\Hcal}) \geq 8n-m + 1$, Properties~\ref{item:1}-\ref{item:6} of $X'$ imply that $X'$ contains at least one base vertex $z$ in an $AB$-copy $\tilde{H}$ of $H$. Let $s$ be the vertex in the pendent copies of $H$ such that $\{z,s\} \in E(G_{\varphi}^{\Hcal})$; hence $\{z,s\}$ is the pendent edge of $\tilde{H}$ (cf. Figure~\ref{fig:exampleC4}). By Property~\ref{item:5} of $X'$, it follows that $z \in X'$. Let $e^{\star} = \{z,s\}$, and let $w$ be the vertex in $G_{\varphi}^{\Hcal} / e^{\star}$ resulting from the contraction of $e^{\star}$. Since both $z,s \in X'$, it can be easily verified that the set $X^{\star} := (X' \setminus \{z,s\}) \cup \{w\}$ is a $\Hcal$-hitting set of $G_{\varphi}^{\Hcal} / e^{\star}$ with $|X^{\star}| = |X'| -1$.
 Therefore,
 $$
\tau_{\Hcal}(G_{\varphi}^{\Hcal} / e^{\star})\  \leq\  |X^{\star}| \ <\ |X'|\ = \ \tau_{\Hcal}(G_{\varphi}^{\Hcal}),
 $$
 and the claim follows.
 \end{proof}

Claims~\ref{claim:likeVC},~\ref{claim:no-edge}, and~\ref{claim:edge} together imply that $\varphi$ is satisfiable if and only if $G_{\varphi}^{\Hcal}$ is a \no-instance of \contracHcalone for the subgraph relation, as we wanted to prove.

\smallskip

Let us now argue that the same proof applies when $\prec$ is another of the graph containment relations stated in the theorem. Indeed, by construction of $G_{\varphi}^{\Hcal}$, the hypothesis that all the graphs in $\Hcal$ are $2$-connected, and the fact that $\Hcal$ is an antichain, it follows that if $X$ is an $H$-hitting set for some of these containment relations, none of the graphs in $\Hcal$ occurs in
$G_{\varphi}^{\Hcal} \setminus X$ nor in $(G_{\varphi}^{\Hcal} / e) \setminus X$ for any edge $e$, for any of the subgraph, induced subgraph, minor, or topological minor containment relations.

\smallskip

Finally, the latter statement in the theorem follows easily from Lemma~\ref{lem:clean3SAT} and by observing that, since the construction of $G_{\varphi}^{\Hcal}$ depends on a {\sl fixed} graph $H \in \Hcal$, it follows that $|V(G_{\varphi}^{\Hcal})| = \Ocal (n)$ and $\Delta(G_{\varphi}^{\Hcal}) \leq 5 \cdot \Delta(H)$.
\end{proof}

From Theorem~\ref{thm:hard-Hcal} we immediately get the following corollary.

\begin{corollary}\label{cor:hard-fvs-oct}
\contracpione is \coNP-hard if $\pi=\fvs$ or $\pi=\oct$.
\end{corollary}
\begin{proof}
For $\pi=\fvs$ (resp. $\pi=\oct$), we apply Theorem~\ref{thm:hard-Hcal} for $\Hcal$ being the collection of all cycles (resp. odd cycles) and $\prec$ being the subgraph relation.
\end{proof}

Note that we can also obtain hardness results assuming that the input graph of the considered problem is planar (and some graph in $\Hcal$ as well), by reducing from planar versions of 3-\textsc{Sat}.

\smallskip

More interesting is the fact the proof of Theorem~\ref{thm:hard-Hcal} does {\sl not} work if either all the graphs in $\Hcal$ are cliques, or if $\Hcal$ contains some graph that is not $2$-connected. Indeed, in the proof of Claim~\ref{claim:no-edge} we crucially used the fact that the vertices $u,v \in V(H)$ are {\sl not} adjacent, so that the contraction of any edge $e$ still leaves intact one of each pair of copies of $H$ in $G_{\varphi}^{\Hcal}/e$. On the other hand, if $\Hcal$ contains a graph $H'$ that is {\sl not} $2$-connected, Claim~\ref{claim:likeVC} does not hold anymore: such a graph $H'$ may occur in the graph $G_{\varphi}^{\Hcal} \setminus X$ considered in the first part of the proof, hence $X$ may not be an  $\Hcal$-hitting set of $G_{\varphi}^{\Hcal}$ anymore.

\smallskip

We now present two hardness results for families $\Hcal$ in which we drop one of the two assumptions discussed above,  namely complete graphs and families containing paths.

 In the next theorem we prove, using a simple trick, \coNP-hardness when $\Hcal$ consists of complete graphs, for the  minor and topological minor containment relations. Note that we may assume that the complete graphs have at least three vertices, as otherwise the problem can be solved in polynomial time by Theorem~\ref{thm:vc-XP}. 

\begin{theorem}\label{thm:hard-cliques}
Let $\Hcal$ be a collection of cliques, each having at least three vertices. Then \contracHcalone is \coNP-hard, for $\prec$ being the minor or topological minor containment relations.
\end{theorem}
\begin{proof}
Let $\Hcal^{\bullet} = \{H^{\bullet} \mid H \in \Hcal\}$, where $H^{\bullet}$ is the graph obtained from $H$ by subdividing each edge once. (If $\Hcal$ is the minor relation, we may assume that $\Hcal$ contains only one clique.) Since all the graphs in $\Hcal$ are cliques on at least three vertices, $\Hcal^{\bullet}$ is a collection of $2$-connected graphs none of which is a clique, hence Theorem~\ref{thm:hard-Hcal} can be applied to it. Given a clean
3-\textsc{Sat} formula $\varphi$, let $G_{\varphi}^{\Hcal^{\bullet}}$ be the graph constructed in the proof of Theorem~\ref{thm:hard-Hcal}  for the family $\Hcal^{\bullet}$. Claims~\ref{claim:likeVC},~\ref{claim:no-edge}, and~\ref{claim:edge} together imply that $\varphi$ is satisfiable if and only if $G_{\varphi}^{\Hcal^{\bullet}}$ is a \no-instance of \contracHcalonebullet for the subgraph relation.
The important observation is that, by construction, in both $G_{\varphi}^{\Hcal^{\bullet}}$ and $G_{\varphi}^{\Hcal^{\bullet}}/e$ for any edge $e$, the gadgets that we attach to a pair of vertices, as well as their proper subgraphs once we delete a $\Hcal^{\bullet}$-hitting set, contain a graph in $\Hcal^{\bullet}$ as a subgraph if and only if they contain one of the cliques in $\Hcal$ as a minor or as a topological minor, and the same proof yields the claimed result.
\end{proof}

Note that the proof of Theorem~\ref{thm:hard-cliques} does not work for the subgraph or induced subgraph containment relations: in that case, the constructed graph $G_{\varphi}^{\Hcal^{\bullet}}$ does not contain any clique of size at least three.

\smallskip

In our next theorem we change appropriately the construction of the graph $G_{\varphi}^{\Hcal}$ defined in the proof of Theorem~\ref{thm:hard-Hcal} to obtain a hardness result when $\Hcal$ consists of a path on at least four vertices and any collection of $2$-connected graphs, for any of the containment relations discussed above.

\begin{theorem}\label{thm:hard-paths}
Let $H=P_i$ with $i \geq 4$, and let $\Hcal$ contain $H$ and any collection of $2$-connected graphs. Then \contracHone is \coNP-hard, for $\prec$ being any of the subgraph, induced subgraph, minor,  or topological minor containment relations.
\end{theorem}
\begin{proof}
We present again a reduction from the 3-\textsc{Sat} problem restricted to clean formulas, similar in spirit to that of Theorem~\ref{thm:hard-Hcal}. Given a clean formula $\varphi$ with $n$ variables and $m$ clauses, we will construct in polynomial time an instance $G_{\varphi}^{\Hcal}$ such that $\varphi$ is satisfiable if and only if $G_{\varphi}^{\Hcal}$ is a \no-instance of \contracHcalone.
We start with the same graph $G_{\varphi}$ defined in the proof of Theorem~\ref{thm:hard-Hcal}, and we modify it as follows. In order to define the graph $G_{\varphi}^{\Hcal}$, we distinguish two cases according to the parity of $i$, the number of vertices in $H=P_i$. In
 this proof, whenever we attach a path, we choose as attachment vertices the endvertices of the path.

  If $i \geq 4$ is even, we replace each of the edges between two vertices in $A$ or two vertices in $B$ by a $P_{\frac{i}{2}+1}$, and we attach a $P_{\frac{i}{2}}$ to each vertex in $A \cup B$.

  If $i \geq 5$ is odd, we replace each of the edges between two vertices in $A$ or two vertices in $B$ by a $P_{\frac{i+1}{2}}$, and we attach a $P_{\frac{i+1}{2}}$ to each vertex in $A \cup B$.

  The remainder of the construction of $G_{\varphi}^{\Hcal}$ is the same for both $i$ even and odd. We replace each edge between a vertex in $A$ and a vertex in $B$ by a $P_{3}$ (note that this does {\sl not} depend on $i$) and, for each such a $P_{3}$, let $z$ be the internal vertex in it. Attach a pendent copy of $H_u^2$ to $z$, where $H=P_i$ and $u$ is one of the endvertices of $P_i$.
This concludes the construction of $G_{\varphi}^{\Hcal}$; see Figure~\ref{fig:examplePi} for examples for $H=P_4$  and $H=P_5$.

\begin{figure}[ht]
\begin{center}
\includegraphics[scale=1.05]{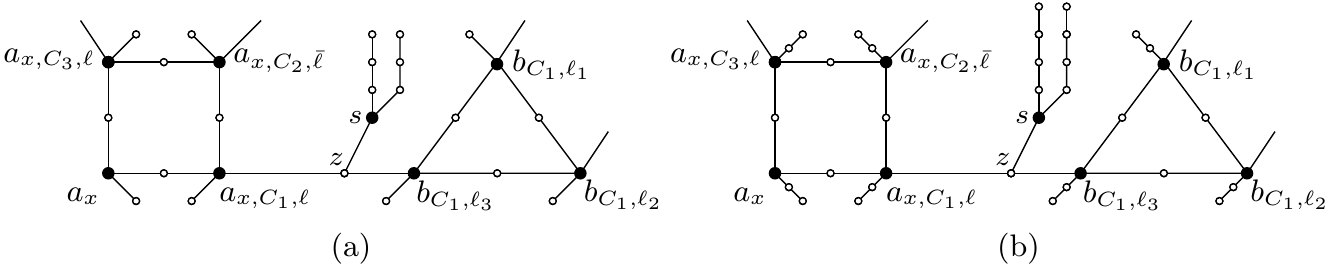}
\end{center}
\caption{Illustration of the graph $G_{\varphi}^{\Hcal}$ for (a) $H=P_4$ and (b) $H=P_5$.}
\label{fig:examplePi}
\end{figure}

 Again, suppose first that $\prec$ is the subgraph relation. The main properties of $G_{\varphi}^{\Hcal}$ are the same as in Theorem~\ref{thm:hard-Hcal}: we may assume that a minimum $P_i$-hitting set $X \subseteq V(G_{\varphi}^{\Hcal})$
  contains exactly one of the pairs $\{ a_{x,C_1,\ell}, a_{x,C_3,\ell} \}$ and $\{ a_x,a_{x,C_2,\bar{\ell}} \}$ for each variable $x$, that $X$ contains exactly $|C|-1$ vertices in the set $\{ b_{C, \ell} \mid \ell \in C\}$ for each clause $C$, and that $X$ contains precisely the attachment vertex of every copy of $H_u^2$ (cf. vertices $s$ in Figure~\ref{fig:examplePi}). Note that since all the graphs in $\Hcal \setminus \{P_i\}$ are $2$-connected by hypothesis, and for any $P_i$-hitting set $X$ every connected component of $G_{\varphi}^{\Hcal} \setminus X$ is a tree, it follows that any $P_i$-hitting set of $G_{\varphi}^{\Hcal}$ is also an $\Hcal$-hitting set of $G_{\varphi}^{\Hcal}$. Moreover, by construction of $G_{\varphi}^{\Hcal}$, these properties are preserved in $G_{\varphi}^{\Hcal} / e$ for any edge $e \in E(G_{\varphi}^{\Hcal})$.

  Taking into account the above discussion, it can be verified that the current graph $G_{\varphi}^{\Hcal}$ satisfies Claims~\ref{claim:likeVC},~\ref{claim:no-edge}, and~\ref{claim:edge} in the proof of Theorem~\ref{thm:hard-Hcal}, by using the same arguments; we omit the details. Thus, $\varphi$ is satisfiable if and only if $G_{\varphi}^{\Hcal}$ is a \no-instance of \contracHcalone for the subgraph relation, as we wanted to prove. As for other containment relations, let $X$ be a $P_i$-hitting set for some of the considered relations.  Since $P_i$ or any $2$-connected graph in $\Hcal$ does not occur in $G_{\varphi}^{\Hcal} \setminus X$ nor in $(G_{\varphi}^{\Hcal}/e) \setminus X$ for any edge $e$,  for any of the subgraph, induced subgraph, minor, or topological minor relations, the same arguments apply.
  \end{proof}

Note that the proof of Theorem~\ref{thm:hard-paths} does {\sl not} work for $H=P_3$. Indeed, in the construction of $G_{\varphi}^{\Hcal}$ for odd $i$, we replace  the edges with both endvertices in $A$ or in $B$ by a $P_{\frac{i+1}{2}}$; for $i=3$ this results in an edge between such a pair, whose contraction would identify both vertices, hence violating the main properties of the reduction.

Finally, note also that, as in Theorem~\ref{thm:hard-Hcal}, the reductions given in
Theorem~\ref{thm:hard-cliques} and Theorem~\ref{thm:hard-paths} also rule out the existence of subexponential algorithms assuming the \ETH.

\section{The case of \textsc{Vertex Cover}}
\label{sec:vertex-cover}

In this section we focus on the case where the considered invariant $\pi$ is the size of a minimum vertex cover or, equivalently, where $\pi = \tau_{K_2}^{\prec}$ for $\prec$ being any of the subgraph, induced subgraph, minor, or topological minor containment relations. Recall that we use the notation $\vc$ to denote  $\tau_{K_2}^{\prec}$.  It is easy to see that \contracvc is \NP-hard, even if we assume that the value $\vc(G)$ is given along with the input. Indeed, the particular case $d=\vc(G)-1$ is the problem of reducing the vertex cover number of the (connected) input graph $G$ to one (i.e., obtaining a star) by doing at most $k$ edge contractions. This problem is known in the literature as \textsc{Star Contraction}~\cite{KrithikaM0T16,HeggernesHLLP14} and is equivalent to \textsc{Connected Vertex Cover} (see~\cite{KrithikaM0T16} for a proof), which is known to be \NP-hard even on graphs for which computing a minimum vertex cover can be done in polynomial time, such as bipartite graphs~\cite{EscoffierGM10}.

 \smallskip

 Following Heggernes et al.~\cite{HeggernesHLP13},  a \emph{$2$-coloring} of a graph $G$ is a function $\phi: V(G) \to \{1,2\}$, and we denote by $V_{\phi}^1$ and $V_{\phi}^2$ the sets of vertices of $V(G)$ colored 1 and 2, respectively. A set $X \subseteq V(G)$ is a \emph{monochromatic component} of $\phi$ if $G[X]$ is a connected component of $G[V_{\phi}^1]$ or $G[V_{\phi}^2]$, and we denote by $\Mcal_{\phi}$ the set of all monochromatic components of $\phi$. The \emph{cost} of a $2$-coloring $\phi$ is defined as ${\sf cost}(\phi) = \sum_{X \in \Mcal_{\phi}} (|X| - 1)$. We will need the following lemma.

\begin{lemma}[Heggernes et al.~\cite{HeggernesHLP13}]
\label{lem:2coloring}
A graph $G$ has a $2$-coloring of cost at most $k$ if and only if there exists a
set $F \subseteq E(G)$ of at most $k$ edges such that $G/F$ is bipartite.
\end{lemma}

The following simple observation concerning the \contracvcone problem is the key insight to the algorithm of Theorem~\ref{thm:vc-XP}. Let $G$ be a graph and let $X$ be a minimum vertex cover of $G$. We define a $2$-coloring $\phi$ of $G$ as follows. For every vertex $v \in V(G)$, $\phi(v) = 1$ if $v \in X$, and
$\phi(v) = 2$ otherwise. Since $X$ is a vertex cover, $G[V_{\phi}^2]$ is edgeless. Consider the graph $G[V_{\phi}^1]=G[X]$, and distinguish two cases according to whether $G$ is bipartite or not. If it is not, then since
$G[V_{\phi}^2]$ is edgeless, necessarily $G[X]$ contains some edge $e$ (equivalently, ${\sf cost}(\phi) \geq 1$). Then contracting $e$ results in a graph having a vertex cover of size at most $|X|-1$, and therefore we can conclude that $G$ is a \yes-instance of the \contracvcone problem. Otherwise, if $G$ is bipartite, we can
solve \contracvcone on $G$ in polynomial time by first computing $\vc(G)$ in polynomial time using the fact that $G$ is bipartite~\cite{Diestel12}, and then computing $\vc(G /e )$ for every edge $e \in E(G)$ in polynomial time as explained below. If for some $e \in E(G)$,  we have that $\vc(G /e ) < \vc(G)$, we answer `\yes', otherwise we answer `\no'.  To compute $\vc(G /e )$ in polynomial time, let $w$ be the vertex resulting from the contraction of $e$ and, letting $G_e := G/e$, note that
$$
\vc(G_e)\ = \ \min\{ 1 + \vc( G_e \setminus \{w\})\ ,\ |N(w)| + \vc(G_e \setminus N[w])      \},
$$
and that both $G_e \setminus \{w\}$ and $G_e \setminus N[w]$ are bipartite, so a minimum vertex cover in them can be computed in polynomial time.

 Summarizing, the algorithm to solve \contracvcone in polynomial time works as follows: we first check whether $G$ is bipartite (in polynomial time). If it is not, we answer `\yes' (without needing to compute any minimum vertex cover). If it is, we solve the problem in polynomial time as discussed above.


\smallskip

In  Theorem~\ref{thm:vc-XP} (cf. Algorithm~\ref{alg:XPvc}) we generalize this idea to solve \contracvcd in polynomial-time for every fixed $d \geq 1$. We first need some technical lemmas.

\begin{lemma}\label{lem:2approxVC}
Let $G$ be a graph, $d \geq 1$ an integer, and $C$ a connected component of $G$ such that $\vc(C) \geq d+1$. Then there exists a set $F \subseteq E(G)$ with $|F| \leq 2d$ such that $\vc(G/F) \leq \vc(G) -d$.
\end{lemma}
\begin{proof}
The main observation is that for any connected graph $H$ such that $\vc(H) \geq 2$, any minimum vertex cover $X$ of $H$ contains two vertices $u,v$ within distance at most two in $H$. Indeed, either $H[X]$ contains an edge, and we choose $u,v$ to be the endvertices of that edge, or since $H$ is connected and $\vc(H) \geq 2$, necessarily there is a vertex in $V(H) \setminus X$ with at least two neighbors in $X$,
 which we choose as $u,v$. In both cases, contracting a shortest path (of length at most two) between such vertices  $u$ and $v$ results in a graph $H'$ with $\vc(H') \leq \vc(H)-1$.

Let $G$, $d$, and $C$ be as in the statement of the lemma. Since $\vc(C) \geq d+1$, we can recursively apply $d$ times the above observation to $C$, hence obtaining a set $F \subseteq E(C) \subseteq E(G)$ of size at most $2d$ such that $\vc(C/F) \leq \vc(C) -d$. Since the size of a minimum vertex cover is additive with respect to connected components, we have that $\vc(G/F) \leq \vc(G) -d$.
\end{proof}

\begin{lemma}\label{lem:smallVC}
Let $G$ be a graph, $d \geq 1$ an integer, and suppose that for every connected component $C$ of $G$, it holds that $\vc(C) \leq d$. Then the \mincontracvc problem with input $(G,d)$ can be solved in time $f(d) \cdot n^{\Ocal(1)}$ for some computable function $f$.
\end{lemma}
\begin{proof}
Let $C_1, \ldots, C_p$ be the connected components of $G$. Since  $\vc(C_i) \leq d$ for $i \in  [p]$, it is easy to observe that $\tw(G) \leq d+1$. For every two integers $i,d'$ with $i \in  [p]$ and $0 \leq d' \leq d$, we apply Proposition~\ref{thm:Courcelle} to solve \mincontracvc with input $(C_i,d')$ in time $f(d) \cdot n$ for some computable function $f$. For this, we just have to verify that the \mincontracvc problem on $(C_i,d')$ can be expressed by an \MSO formula whose length depends only on $d$. Indeed, it consists of finding the minimum size of a set $F \subseteq E(C_i)$ such that $\vc(C_i / F) \leq \vc(C_i) -d'$. To express the latter inequality by an \MSO formula with length depending on $d$, we crucially use the assumption that $\vc(C_i) \leq d$. To do this, we first compute $\ell:=\vc(C_i)$ independently with a standard \MSO formula (or with a standard branching algorithm, since we are assuming that $\vc(C_i) \leq d$). Then the inequality ``$\vc(C_i / F) \leq \ell-d'$'', where we have that $\ell-d' \leq d$,  can be expressed as the existence of a set of vertices $S:=\{v_1,\ldots,v_{\ell -d'} \}\subseteq V(C_i)$ such that every edge in $E(C_i) \setminus F$ has an endpoint in $S$ or has an endpoint $u \in V(F)$ such that there exists a vertex $v \in S$ and a path from $u$ to $v$ in $C_i$ using only edges in $F$. (This latter case captures the fact that an edge $e$ of $C_i / F$ can also be covered by a vertex $v \in V(F)$ that becomes eventually an endpoint of $e$ after contracting the edges in $F$.)

\smallskip


Let ${\sf opt}(C_i,d')$ be the output of \mincontracvc with input $(C_i,d')$, for $i \in  [p]$ and $0 \leq d' \leq d$. We assume that ${\sf opt}(C_i,d') = \infty$ if $\vc(C_i) \leq d'$, and ${\sf opt}(C_i,0) = 0$. With this information at hand, we present a simple dynamic programming algorithm to solve the \mincontracvc problem with input $(G,d)$ within the claimed running time.

Let ${\sf dp}(i,j)$ be the minimum size of a set $F \subseteq E(C_1) \cup \ldots \cup E(C_i)$ such that $\vc(G/F) \leq \vc(G) - j$, or $\infty$ if such set does not exist. Note that, in order to compute ${\sf dp}(i,j)$, if in an optimal solution the size of a minimum vertex cover drops by $q$ in $C_i$, then ${\sf dp}(i,j) = {\sf dp}(i-1,j-q) + {\sf opt}(C_i,q)$.
Then ${\sf dp}(i,j)$ can be computed as follows.

$$
{\sf dp}(i,j) = \left\{
        \begin{array}{ll}
            0 					& \mbox{if } j = 0, \\
            \infty 				& \mbox{if }i = 0 \mbox{ and } j > 0, \\
            \displaystyle \min_{0 \leq q \leq j} {\sf dp}(i-1,j-q)+ {\sf opt}(C_i,q) 		 & \mbox{otherwise.}
        \end{array}
    \right.
$$

Note that each ${\sf dp}(i,j)$ can be computed in time $\Ocal(j)$. Recall that $p$ is the number of connected components of $G$. Hence, since $p \leq n$,  ${\sf dp}(i,j)$ can be computed for each pair $i,j$ in total time $\Ocal(n\cdot d^2)$ and the answer is given by ${\sf dp}(p,d)$.\end{proof}

The \emph{bipartite contraction number} of a graph $G$, denoted by $\bc(G)$, is the minimum size of a set $F \subseteq E(G)$ such that $G/F$ is bipartite. We will use the following result of Heggernes et al.~\cite{HeggernesHLP13} as a subroutine in our algorithms. Note that~\cite{GUILLEMOTMARX} presents a faster algorithm.

\begin{proposition}[Heggernes et al.~\cite{HeggernesHLP13}]
\label{thm:bcFPT}
Given a graph $G$ and a positive integer $k$, deciding whether $\bc(G) \leq k$ is \FPT parameterized by $k$.
\end{proposition}

We finally have all the ingredients to present our main algorithm.

\begin{theorem}\label{thm:vc-XP}
The \contracvc problem is in \XP parameterized by $d$. In particular, \contracvcd is polynomial-time solvable for every fixed $d \geq 1$.
\end{theorem}
\begin{proof}
Let $(G,k,d)$ be the input of \contracvc, and let $n = |V(G)|$. The  \XP algorithm that we proceed to present is summarized in Algorithm~\ref{alg:XPvc}.

\begin{algorithm}[t]
\LinesNumberedHidden
\KwIn{A triple $(G,k,d)$ with $n = |V(G)|$.}

\BlankLine

\If{$k < d$}{\Return{\no}.}
\Else(\text{($k \geq d$)}){\text{Check whether $\bc(G) \leq d-1$ in time $f(d) \cdot n^{\Ocal(1)}$ by Proposition~\ref{thm:bcFPT}.}

\If{$\bc(G) \geq d$}{\Return{\yes}.}

\Else(\text{($\bc(G) < d$)}){\text{Let $C_1, \ldots, C_p$ be the connected components of $G$.}\\
\text{For $i \in [p]$, check whether $\vc(C_i) \leq d$ in time $2^{\Ocal(d)} \cdot n^{\Ocal(1)}$.}\\
\If{\text{$\vc(C_i) \leq d$ \emph{for every} $i \in [p]$}}{
Solve \mincontracvc with input $(G,d)$  in time $f(d) \cdot n^{\Ocal(1)}$ by Lemma~\ref{lem:smallVC}. Let $k_0$ be the optimal solution.\\
\If{$k \leq k_0$}{\Return{\yes}.}
\Else(\text{($k > k_0$)}){\Return{\no}.}
}

\Else(\text{(there is a component $C$ with $\vc(C) \geq d+1$)}){
\If{$k \geq 2d$}{\Return{\yes} by Lemma~\ref{lem:2approxVC}.}
\Else(\text{($k < 2d$)}){
Enumerate all sets $F \subseteq E(G)$ with $|F| \leq k \leq 2d-1$ in time $n^{\Ocal(d)}$.\\
For each $F$, compute $\vc(G/F)$ in time $2^{\Ocal(d)} \cdot n^{\Ocal(1)}$.\\
\If{\emph{for some $F$, $\vc(G/F) \leq \vc(G) - d$}}{\Return{\yes}.}
\Else(\text{(there is no $F$ such that $\vc(G/F) \leq \vc(G) - d$)}){\Return{\no}.}
}
}
}
}
\caption{\XP algorithm for the \contracvc problem parameterized by $d$.} \label{alg:XPvc}
\end{algorithm}

Note that since the contraction of an edge may drop the minimum vertex cover of a graph by at most one, we may assume that $k \geq d$, as otherwise the answer is trivially `\no'. We start by checking whether $\bc(G) \leq d-1$ by using Proposition~\ref{thm:bcFPT} in time $f(d) \cdot n^{\Ocal(1)}$. We distinguish two cases.

Assume first that $\bc(G) \geq d$, and let $X$ be a minimum vertex cover of $G$, which is only used for the analysis. We define a $2$-coloring $\phi$ of $G$ as follows. For every vertex $v \in V(G)$, $\phi(v) = 1$ if $v \in X$, and $\phi(v) = 2$ otherwise. Since $X$ is a vertex cover, $G[V_{\phi}^2]$ is edgeless. Since $\bc(G) \geq d$, Lemma~\ref{lem:2coloring} implies that ${\sf cost}(\phi) \geq d$, which in turn implies, since $G[V_{\phi}^2]$ is edgeless, that $G[V_{\phi}^1]=G[X]$ contains a set of connected components $\Mcal$ such that  $\sum_{X \in \Mcal} (|X| - 1) \geq d$. Then contracting in $G$ any set $F$ of $d$ edges of a spanning forest of $\Mcal$ results in a graph $G/F$ such that $\vc(G/F) \leq \vc(G) -d$. Since we may assume that $k \geq d$, in this case we can safely answer `\yes'.

Otherwise, we have that $\bc(G) \leq d-1$.
Let $C_1, \ldots, C_p$ be the connected components of $G$.
For every $i \in [p]$, we check whether $\vc(C_i) \leq d$ in time $2^{\Ocal(d)} \cdot n^{\Ocal(1)}$ by using an \FPT algorithm for \textsc{Vertex Cover}~\cite{CyganFKLMPPS15}. We distinguish again two cases.

If $\vc(C_i) \leq d$ for every $i \in [p]$, we apply Lemma~\ref{lem:smallVC} and solve the \mincontracvc problem with input $(G,d)$  in time $f(d) \cdot n^{\Ocal(1)}$ for some computable function $f$. If the optimal solution is larger than $k$, we answer `\no', otherwise we answer `\yes'.

Otherwise, there exists a connected component $C$ of $G$ such that
$\vc(C) \geq d+1$. By Lemma~\ref{lem:2approxVC}, there exists a set $F \subseteq E(G)$ with $|F| \leq 2d$ such that $\vc(G/F) \leq \vc(G) -d$. Hence, if $k \geq 2d$, we answer `\yes'. Otherwise, we have that $k \leq 2d-1$, and we solve the problem in time $n^{\Ocal(d)}$ as follows. We enumerate all candidate sets $F \subseteq E(G)$ with $|F| \leq k \leq 2d-1$, which are $n^{\Ocal(d)}$ many, and for each such a set $F$, compute $\vc(G/F)$ in time $2^{\Ocal(d)} \cdot n^{\Ocal(1)}$ as explained below. With the same technique we are also able to compute $\vc(G)$ within the same runtime bound. If for some such a set $F$, we have that $\vc(G/F) \leq \vc(G) - d$, we answer `\yes', otherwise we answer `\no'.

Let us now see, given a set $F \subseteq E(G)$ with $|F| \leq k \leq 2d-1$, how $\vc(G/F)$ can be computed in time $2^{\Ocal(d)} \cdot n^{\Ocal(1)}$. To do that, we start by finding a set $B \subseteq V(G/F)$ with $|B| = \Ocal(d)$ such that $(G/F) \setminus B$ is bipartite, as explained in the next paragraph. Once we have the set $B$ at hand, we can  guess which vertices of $B$ belong to the vertex cover. Since $|B|=\Ocal(d)$, this can be done within the claimed running time. We can now delete the vertices of $B$ from $G/F$, together with $N(v)$ for every $v\in B$ such that $v$ belongs to the vertex cover. The remaining graph is bipartite, hence we can compute a minimum vertex cover of it in polynomial time. Note that this procedure works even if $F=\emptyset$, hence $\vc(G)$ can also be computed in time $2^{\Ocal(d)} \cdot n^{\Ocal(1)}$. Therefore, to conclude the proof if it enough to find such a set $B \subseteq V(G/F)$.


Recall that we are in the case where $\bc(G) \leq d-1$. Let $L \subseteq E(G)$ with $|L| \leq d-1$ such that $G/L$ is bipartite, obtained in time \FPT in $d$ by Proposition~\ref{thm:bcFPT} (it is easy to see that the  \FPT algorithm for the decision version can also obtain  in \FPT time the corresponding set of edges to be contracted). Note that $G \setminus V(L)$ is also bipartite. Let $V_F$ be the set of vertices in $V(G/F)$ resulting from the contraction of $F$. We set $B:=V(L) \cup V_F$. Note that $|B| \leq  |V(L)| + |V_F| \leq 2(d-1) + 2(2d-1) = \Ocal(d)$ and that $(G/F) \setminus B$ is a subgraph of $G \setminus V(L)$, hence it is bipartite as well, and the theorem follows.
\end{proof}


From the \XP algorithm given in Theorem~\ref{thm:vc-XP} we easily get the following corollary.

\begin{corollary}\label{cor:vc-2-approx}
The \mincontracvc problem can be 2-approximated in \FPT time parameterized by $d$.
\end{corollary}
\begin{proof}
Let $(G,d)$ be the input of the \mincontracvc problem, and let $k_0$ be the desired minimum number of contractions to drop the vertex cover number of $G$ by at least $d$, so necessarily $k_0 \geq d$. Note that the algorithm given in Theorem~\ref{thm:vc-XP} (cf.~Algorithm~\ref{alg:XPvc}) either concludes that there is no feasible solution (in case $k_0<d$), or concludes that $k_0 = d$ (in case $\bc(G) < d$), or solves optimally the \mincontracvc problem with input $(G,d)$ (in case $\vc(C_i) \leq d$ for every connected component $C_i$ of $G$), except in the case where there exists a component $C$ of $G$ such that $\vc(C) \geq d+1$. That is, this algorithm in fact solves the \mincontracvc problem with input $(G,d)$ except for this latter case, which is in fact the only step of the algorithm that is not \FPT in $d$, since we need to enumerate the candidate sets $F \subseteq E(G)$ of size at most $2d-1$. But if we are aiming at a $2$-approximation algorithm, in the case where there is a component $C$ of $G$ with $\vc(C) \geq d+1$, we can just apply Lemma~\ref{lem:2approxVC} directly and conclude that $k_0 \leq 2d$. Since $k_0 \geq d$, this yields a 2-approximation that runs in time \FPT parameterized by $d$.
\end{proof}

\section{Conclusions and further research}
\label{sec:concl}

We provided \coNP-hardness results for the \contracHcalone problem when $\Hcal$ contains only $2$-connected graphs and at least one of them is not a clique (Theorem~\ref{thm:hard-Hcal}), when $\Hcal$ consists of cliques but only for the minor and topological minor containment relations (Theorem~\ref{thm:hard-cliques}), and when $\Hcal$ contains a path on at least four vertices and $2$-connected graphs (Theorem~\ref{thm:hard-paths}). Several interesting cases remain open, for instance when $\Hcal = \{H\}$ with $H=P_3$, $H=K_h$ with $h\geq 3$ (for the subgraph and induced subgraph relations), or $H$ being an arbitrary tree.
The cases where $\Hcal$ may contain disconnected graphs seem to be trickier.

For the cases that are \coNP-hard, it is natural to parameterize the problem by $\tau_{\Hcal}^{\prec}$, that is, by the value of the parameter in the input graph $G$. If $\prec$ is the minor relation and $\Hcal$ contains some planar graph, it is well-known~\cite{RobertsonS86} that the treewidth of $G$ is bounded by $\tau_{\Hcal}^{\prec}(G)$ plus a function that depends only on $\Hcal$. In this case, the \contracHcal problem is \FPT parameterized by $\tau_{\Hcal}^{\prec}(G)+k$, since it can be expressed by an \MSO formula with length depending only on $k$ (note that we may assume that $k \geq d$), and therefore it can be solved in time $f(\tau_{\Hcal}^{\prec}(G),k) \cdot n$ by Courcelle's Theorem~\cite{Courcelle90}. In particular, this observation yields that when $\Hcal=\{K_3\}$, the \contracfvs problem is \FPT parameterized by $\fvs + k$.

\smallskip


When $\Hcal=\{K_2\}$, that is, when $\tau_{\Hcal}^{\prec}$ is the size of a minimum vertex cover, we proved that the \contracvc problem parameterized by $d$ is in \XP (Theorem~\ref{thm:vc-XP}) and
can be 2-approximated in \FPT time (Corollary~\ref{cor:vc-2-approx}). The natural question is whether  \contracvc is \FPT or {\sf W}[1]-hard parameterized by $d$.
We tend to believe that the former case holds. In fact, we do not even know whether the problem is \FPT parameterized by $k$ (recall that we may assume that $k \geq d$).

\medskip

\noindent\textbf{Acknowledgement}. We would like to thank the anonymous reviewers of the conference and journal versions of this article for helpful remarks that improved the presentation of the manuscript.




\bibliography{bibliography}

\end{document}